\documentclass[11pt,a4paper]{article}
\usepackage{amsfonts,amssymb}
\usepackage{cite}
\usepackage[T2A]{fontenc}
\usepackage[cp1251]{inputenc}
\usepackage[english]{babel}
\usepackage{amsmath}
\usepackage{amsthm}
\usepackage[unicode]{hyperref}
\font\Sets=msbm10

\def\Z {\hbox{\Sets Z}}
\DeclareMathOperator{\dn}{\mathrm{dn}}
\DeclareMathOperator{\sn}{\mathrm{sn}}
\DeclareMathOperator{\cn}{\mathrm{cn}}

\usepackage{graphicx}
\graphicspath{}
\DeclareGraphicsExtensions{.pdf,.png,.jpg,.jpeg}

\newtheorem{theorem}{Theorem}[section]
\newtheorem{definition}{Definition}[section]
\newtheorem{proposition}{Proposition}[section]
\newtheorem{example}{Example}[section]

\newtheorem{remark}{Remark}[section]

\numberwithin{equation}{section}

\title{\bf Construction of exact solutions to the Ruijsenaars-Toda \\ lattice via generalized invariant manifolds}% Указываем название статьи
\author{\bf I.T. Habibullin$^{1,}$\footnote{Corresponding author. E-mail address: habibullinismagil@gmail.com (I.T. Habibullin).} , А.R. Khakimova$^{1}$, A.O. Smirnov$^{2}$\\  \\
\scriptsize $^{1}$Institute of Mathematics, Ufa Federal Research Centre, Russian Academy of Sciences, Ufa 450008, Russia\\
\scriptsize $^{2}$Saint-Petersburg State University of Aerospace Instrumentation, St. Petersburg 190000, Russia}
\date{}

\allowdisplaybreaks
\begin{document}

\maketitle

\abstract{The article discusses a new method for constructing algebro-geometric solutions of nonlinear integrable lattices, based on the concept of a generalized invariant manifold (GIM). In contrast to the finite-gap integration method, instead of the eigenfunctions of the Lax operators, we use a joint solution of the linearized equation and GIM. This makes it possible to derive Dubrovin type equations not only in the time variable $t$, but also in the spatial discrete variable $n$. We illustrate the efficiency of the method using the Ruijsenaars-Toda lattice as an example, for which we have derived a real and bounded  particular solution in the form of a kink, a periodic solution expressed in terms of Jacobi elliptic functions and a solution expressed through Weierstrass $\wp$-function.}

\large
%\normalsize

\section{Introduction}

The article deals with the problem of describing particular solutions of the nonlinear integrable lattices. Here, instead of the standard Lax pair, we use an object called a generalized invariant manifold as the main research tool. We take such a classical model as the relativistic Toda lattice (or Ruijsenaars-Toda lattice) \cite{Ruijsenaars} as basic object. 

Note that nonlinear integrable partial differential equations and their various discrete analogs are currently in demand and are being actively studied. The reader can get acquainted with the applications of the integrable lattices in theoretical physics in more detail in the monograph \cite{Arutyunov}.

It is well known that a characteristic feature of integrability is that the equation has an infinite hierarchy of higher symmetries represented in the form of equations of evolutionary type commuting with the equation under consideration. This circumstance underlies an effective symmetry approach to the problem of complete description and classification of integrable equations \cite{MikhailovShabatSokolov}. On the other hand, symmetries make it possible to construct wide classes of explicit particular solutions of integrable equations. The point is that stationary reductions of symmetries of a given equation are ordinary differential (difference) equations compatible with this equation. Moreover, as examples show, they are Hamiltonian systems with complete sets of integrals in involution, which ensures their complete integrability (see, for instance, \cite{DubrovinMatveevNovikov, Veselov, MEKL, Belokolos, ItsMatveev, Batchenko, KacMoerbeke, Zhao, ZhouJiang, GengWuCao, CaoGengWang}).

The integrable equations have a Lax pair, i.e. a pair of linear differential (difference) operators that commute if and only if the function included in the coefficients of the operators satisfies the considered equation (see, for instance, \cite{Zakharov}). The solutions of stationary reductions of higher symmetries of the equation are closely related to the spectral characteristics of the Lax operators. Namely, such solutions correspond to Lax operators, the spectra of which have only a finite number of forbidden zones (lacunae). The spectral curve defines a two-sheeted Riemann surface of finite genus. The eigenfunctions of the Lax operators are expressed in terms of the Baker-Akhiezer function, which is defined and meromorphic on the Riemann surface with an essential singularity of a given type (see survey \cite{DubrovinMatveevNovikov}).

As an illustrative example, let us give some fragments of the article \cite{DubrovinMatveevNovikov} concerning the method of constructing finite-gape solutions of the KdV equation $u_t=6uu_x-u_{xxx}$. The finite-gape solution $u(x,t)$ is expressed in the form 
\begin{align}\label{u_kdv}
u=-2\sum_{j=1}^n\gamma_j+\sum_{i=1}^{2n+1}E_i,
\end{align}
where $E_1<E_2<\ldots<E_{2n+1}$ are arbitrary real constants. Functions $\gamma_1,\gamma_2,\ldots,\gamma_n$ satisfy a consistent pair of systems of ordinary differential equations, called the Dubrovin equations: 
\begin{align}\label{gammax_kdv}
\frac{d}{dx}\gamma_j=\pm \frac{2i\sqrt{P_{2n+1}(\gamma_j)}}{\prod_{k\neq j}\left(\gamma_k-\gamma_j\right)}, \quad i^2=-1,
\end{align}
\begin{align}\label{gammat_kdv}
\frac{d}{dt}\gamma_j=\pm \frac{\left(8i\sum_{s\neq j}\gamma_j-4i\sum_{j=1}^{2n+1}E_j\right)\sqrt{P_{2n+1}(\gamma_j)}}{\prod_{k\neq j}(\gamma_k-\gamma_j)},
\end{align}
where $P_{2n+1}(E)=\prod_{j=1}^{2n+1}\left(E-E_j\right)$.
Thus, formulas \eqref{u_kdv}-\eqref{gammat_kdv} reduce the problem of constructing a solution to a partial differential equation to a simpler problem of integrating ordinary differential equations.

The finite-gape integration method is effectively applied to nonlinear integrable chains as well. Finite-gap solutions of nonlinear Volterra and Toda lattices, as well as the relativistic Toda lattice, were investigated in a number of works \cite{DubrovinMatveevNovikov, Krichever85, Veselov, Ver, Khasanov, Smirnov20, SmirnovAO, Zhao, GongGeng}. In them one can find explicit formulas for the solutions of these chains in terms of the Riemann theta functions, as well as analogues of Dubrovin's equations \eqref{gammat_kdv}, which determine the dynamics in time $t$. As far as the authors know, analogs of the Dubrovin equations describing the dynamics in the spatial variable $n\in Z$ have not been presented before. In \cite{DubrovinMatveevNovikov} it is shown that to obtain such equations, it is necessary to find an explicit solution to a very complex system of higher order algebraic equations.

In the work \cite{HKP2016JPA}, we introduced the concept of a generalized invariant manifold (GIM) of a nonlinear differential (difference, differential-difference) equation, which turned out to be useful in the study of nonlinear equations. It was shown in \cite{HK2020JPA}, \cite{HKS2021UMJ} that the GIM can be effectively used to construct particular solutions of integrable equations. In \cite{HK2020JPA}, using the GIM, analogs of equations \eqref{gammax_kdv} and \eqref{gammat_kdv} for the Volterra chain were derived. However, they are of the second order, and not the first, as should be in the Dubrovin equations. In this paper, we improve the approach proposed in \cite{HK2020JPA} and select the GIM in such a way that the equations of Dubrovin type that follow from it have the correct structure. As a result, we derived discrete analogs of the Dubrovin equations for the relativistic Toda lattice (see \S5 formulas \eqref{dotgj}, \eqref{gjnp1}). It worth mentioning that equation \eqref{gjnp1} defines the B$\ddot{a}$cklund transformation for equation \eqref{dotgj}. Ordinary discrete equations of this kind are actively studied in the framework of the theory of integrable mappings (see, for example, \cite{Veselov2, Joshi, Nijhoff}).  

Let us explain the essence of the notion of generalized invariant manifold. First we recall a notion of the invariant manifold.  The invariant manifold method is a widespread method for constructing particular solutions of nonlinear partial differential equations (see, for example, \cite{Yanenko}). This method is based on the following technique. Some simpler equation, is assigned to the investigated equation, that is assumed to be compatible with it. This equation is called invariant manifold or differential constraint, some authors use the term conditional symmetry (see \cite{Zhdanov}, \cite{Sergeyev}). The consistency of the two equations ensures that a solution to the added equation can be found, which will also be a solution to the original one. It remains only to find such an effectively solvable equation for a given nonlinear equation. Generally speaking, this is a difficult problem. One of the effective tools for searching invariant manifolds is based on the classical and higher symmetries  \cite{MikhailovShabatSokolov, Olver}.  In \cite{FokasLiu}, \cite{Zhdanov} the notion of generalized conditional symmetry was introduced providing a method for finding invariant manifolds and exact solutions of the nonlinear partial differential equations (see also \cite{Sergeyev}). An effective analytical method for finding invariant manifolds can be found in \cite {Demina}. Note that in this work, the Puiseux series method was used to describe the algebraic invariants of the Kuramoto -- Sivashinsky dispersion equation. The found invariant manifolds allowed the author to construct new exact solutions. Now we turn back to our generalization. The essence of the concept of a generalized invariant manifold is that we add an invariant manifold not to the equation itself, but to its linearization. The exact definition of GIM can be found in \S2 (Definition 2.1). In contrast to usual invariant manifold GIM is effectively found (see \S2, comment after Definition 2.1). 

It should be noted that the class of generalized invariant manifolds is quite extensive. In examples \ref{Ex1} and \ref{Ex2}, we consider generalized invariant manifolds with the sought function $U_n$ or respectively, a pair of sought functions $U_n$, $V_n$ of small orders, assuming that variables $u_n, u_{n\pm 1}, \ldots$, as well as variables $v_n, v_{n\pm 1}, \ldots$ enter the equation as functional parameters. We remark that these equations are solved explicitly, their solutions coincide with the generators of classical symmetries.

In more complicated examples, GIMs contain arbitrary constants parameters. Generally speaking, these GIMs do not admit explicit formulas for the general solution, but they admit formal solutions in the form of asymptotic series in the neighborhood of singular values of the parameters. It is remarkable that the termination of the asymptotic series generates ordinary differential (difference) equations, which are actually ordinary invariant manifolds for the nonlinear equation under consideration. In other words, generalized invariant manifolds can be used to construct ordinary invariant manifolds suitable for constructing particular solutions of the equation (see \eqref{fN} and \eqref{gN}).

We note that a generalized invariant manifold differs significantly from the concept of generalized conditional symmetry defined in \cite{Zhdanov}, \cite{FokasLiu}. However, as noted in \cite{ZYZhang},  GIM can be interpreted as a generalized conditional symmetry for a system composed of a given equation and its linearization.

Let's briefly discuss how the article is organized. In \S2, we formulate a rigorous definition of the concept of a generalized invariant manifold for nonlinear lattices, give illustrative examples, and discuss potential applications. In section 3, we construct the simplest meaningful example of a second-order generalized invariant manifold for the Ruijsenaars-Toda lattice. It is noteworthy that it contains a dependence on two arbitrary constant parameters. The detailed derivation of this invariant manifold is simple, but rather cumbersome, so we brought it into the Appendix. The linearized equation and the generalized invariant manifold generate a system of ordinary differential equations \eqref{RTL_oim} and a system of ordinary discrete equations \eqref{RTL_lin}. These two systems are compatible. In Section 4, we construct a formal asymptotic expansion in powers of the parameter  $\lambda$ of solutions  of \eqref{RTL_oim} and \eqref{RTL_lin}. The assumption that the formal asymptotic series terminates generates additional ordinary differential equations for the desired solution of the original nonlinear lattice, i.e. they define the usual invariant manifolds \eqref{fN} and \eqref{gN}. The equations obtained are rather complicated, although in simple cases they can be solved. The situation is somewhat simplified when we initially assume that the given solutions of equations \eqref{RTL_oim} and \eqref{RTL_lin} are polynomials on $\lambda$ parametrized by their roots (see Section 5). In this case, we obtain differential and discrete equations on the roots of the polynomials, which are analogs of the Dubrovin equations. We illustrate the effectiveness of this approach in Section 6 by constructing exact solutions of the lattice. Here we find a simple regular solution in the form of a kink, a periodic solution with the period 2, expressed in terms of the elliptic Jacobi functions and a solution expressed through Weierstrass $\wp$-function, the last two solutions are not regular. As far as we know, there are no examples of real and bounded quasi periodic solutions of the Ruijsenaars-Toda lattice.

\section{Preliminaries}

Let us briefly recall the main stages of the method of invariant manifolds using the example of differential-difference equations (lattices)
\begin{align}
\frac{d u_n}{dt}=f(u_{n+k}, u_{n+k-1},\dots ,u_{n-k}), \label{main2}
\end{align}
where the sought function $u=u_n(t)$ depends on a continuous variable $t\in \textbf{R}$ and a discrete variable $n\in \Z$.
A difference equation of the form 
\begin{align}
u_{n+s}=g(u_{n+s-1},u_{n+s-2},\dots,u_{n}) \label{inv}
\end{align}
defines an invariant manifold for the lattice (\ref{main2}) if the following condition is satisfied
\begin{align}
\left.\frac{d}{dt}g-D^s_nf\right|_{\eqref{main2},\eqref{inv}}=0, \label{inv-cond} 
\end{align}
where it is assumed that the derivative $\frac{d g} {du_n}$ does not vanish identically. The shift operator $ D_n $ acts according to the rule $ D_nu_n = u_ {n + 1} $.

Due to the compatibility condition a solution $u_n=u_n(t)$ of the equation \eqref{main2} satisfying relation \eqref{inv} at some moment of time  $t=t_0$ will satisfy the relation for all values of $t$.  Actually this means the invariance of the equation \eqref{inv} with respect to the equation \eqref{main2}.

Obviously relation \eqref{inv-cond} is nothing but a nonlinear differential-difference equation with unknown $g$. Sometimes this equation can be solved explicitly, but as a rule, the problem of finding the function $ g $ is extremely difficult. The situation changes noticeable if we look for an ordinary differential equation that is compatible with the linearization of the nonlinear equation \eqref{main2} at the vicinity of its arbitrary solution.

Let us recall the precise definition of the generalized invariant manifold for equation (\ref{main2}). First we evaluate the linearization (Fr$\acute{e}$chet derivative) of the equation:
\begin{align}
\frac{d U_n}{dt}=\frac{\partial f}{\partial u_{n+k}}U_{n+k}+\frac{\partial f}{\partial u_{n+k-1}}U_{n+k-1}+\ldots+\frac{\partial f}{\partial u_{n-k}}U_{n-k}. \label{main-lin} 
\end{align}
For the linearized equation \eqref{main-lin} function $U_n=U_n(t)$ is a sought function while an arbitrary solution $u_n=u_n(t)$ to the equation \eqref{main2} is considered as a functional parameter.

\begin{definition} A surface in the space of the dynamical variables, defined by the equation 
\begin{align}
U_{n+j}=F(U_{n+j-1},U_{n+j-2},\ldots,U_{n},u_{n},u_{n\pm 1},\ldots,u_{n\pm i}), \quad \frac{\partial F}{\partial U_n}\neq 0\label{main-inv}
\end{align}
is called a generalized invariant manifold for equation (\ref{main2}), if the condition holds
\begin{align}
\left.\frac{d}{dt}F-D_n^j\frac{d U_n}{dt} \right|_{\eqref{main2},\eqref{main-lin},\eqref{main-inv}}=0. \label{main-cond}
\end{align}
\end{definition}
In the latter equation derivatives of the variables $U_{l}$ и $u_{l}$ with respect to $t$ should be excluded due to equations (\ref{main2}) and (\ref{main-lin}). Similarly variables $U_{l}$ for $l<n$ and for $l\geq n+j$ are replaced by means of the equations (\ref{main-inv}). In the resulting equation, the variables $ U_ {n + j-1} $, $ U_ {n + j-2} $, \ldots, $ U_ {n} $ and $ u_l $ for any integer $ l $ are considered as independent. Finally, we arrive at an overdetermined system of equations that can be effectively solved. Below we give an illustrative example of a generalized invariant manifold for the Volterra equation.

\begin{remark}
If worth mentioning that for any solution of the system \eqref{main-lin}, \eqref{main-inv} the relation 
\begin{equation}\label{utauU}
u_{n,\tau}=U_n
\end{equation}
defines a symmetry of equation \eqref{main2}. If $U_n$ is explicitly expressed in terms of the variables $n, t, u_n, u_{n\pm1}, u_{n\pm2}, \ldots$ then \eqref{utauU} determines a classical or higher symmetry, otherwise the symmetry is non-local.
\end{remark}

\begin{example}\label{Ex1}
For the Volterra chain \cite{Volterra}
\begin{equation}\label{volterra}
u_{n,t}=u_{n}(u_{n+1}-u_{n-1}),
\end{equation}
where the sought function $u=u_n(t)$ depends on real $t$ and integer $n$. 
Equation
\begin{equation}\label{troivol}
U_{n+1}=\frac{u_{n+1}+tu_{n+1,t}}{u_{n}+tu_{n,t}}U_n
\end{equation}
defines a generalized invariant manifold, since as it is easily verified equation \eqref{troivol} is compatible with the linearized equation 
\begin{equation}\label{volterralin}
U_{n,t}=u_n(U_{n+1}-U_{n-1})+(u_{n+1}-u_{n-1})U_n.
\end{equation}
The same is true for the equation 
\begin{equation}\label{troi2vol}
U_{n+2}-\frac{u_{n}u_{n+2,t}-u_{n,t}u_{n+2}}{u_{n}u_{n+1,t}-u_{n,t}u_{n+1}}U_{n+1}+ \frac{u_{n+1}u_{n+2,t}-u_{n+1,t}u_{n+2}}{u_{n}u_{n+1,t}-u_{n,t}u_{n+1}}U_{n}=0.
\end{equation}

It is easily checked that equations \eqref{troivol} and \eqref{troi2vol} are solved in a closed form, their general solutions are  $U_n=c(u_{n}+tu_{n,t})$ and, respectively, $U_n=c_1(u_{n}+tu_{n,t})+c_2u_{n,t}$. They correspond to the classical symmetries
$$u_{n,\tau}=c(u_{n}+tu_{n,t}) \quad \mbox{and}\quad u_{n,\tau}=c_1(u_{n}+tu_{n,t})+c_2u_{n,t}.$$
By taking the stationary reductions of the second of the symmetries we get an exact solution for the chain \eqref{volterra}
$$u_n(t)=\frac{a+b(-1)^n+kn}{2kt+b},$$
where $a$, $b>0$, $k>0$ are arbitrary constants.
\end{example}

The notion of GIM is easily adopted to the case of system of the differential-difference equations. Formulas \eqref{main-lin}-\eqref{main-cond} can be rewritten for the case of systems. For greater clarity, we present the definition of GIM for a system of two equations with two unknowns. Consider a system of the form
\begin{equation}\label{sys0}
\begin{aligned}
\frac{d u_n}{dt}=f(u_{n+k},v_{n+k}, u_{n+k-1},v_{n+k-1},\dots ,u_{n-k},v_{n-k}),\\
\frac{d v_n}{dt}=g(u_{n+k},v_{n+k}, u_{n+k-1},v_{n+k-1},\dots ,u_{n-k},v_{n-k}),
\end{aligned} 
\end{equation}
where the sought functions $u=u_n(t)$, $v=v_n(t)$ depend on a continuous variable $t\in \textbf{R}$ and a discrete variable $n\in \Z$. The linearization of \eqref{sys0} is given by the following system of equations
\begin{equation}\label{lin_sys0}
\begin{aligned}
\frac{d U_n}{dt}=\frac{\partial f}{\partial u_{n+k}}U_{n+k}+\frac{\partial f}{\partial v_{n+k}}V_{n+k}+\ldots+\frac{\partial f}{\partial u_{n-k}}U_{n-k}+\frac{\partial f}{\partial v_{n-k}}V_{n-k},\\
\frac{d V_n}{dt}=\frac{\partial g}{\partial u_{n+k}}U_{n+k}+\frac{\partial g}{\partial v_{n+k}}V_{n+k}+\ldots+\frac{\partial g}{\partial u_{n-k}}U_{n-k}+\frac{\partial g}{\partial v_{n-k}}V_{n-k}.
\end{aligned} 
\end{equation}
Here $U_n=U_n(t)$, $V_n=V_n(t)$ are the sought functions, and an arbitrary solution $u=u_n(t)$, $v=v_n(t)$ of the system \eqref{sys0} is considered as a functional parameter.
\begin{definition} A surface in the space of the dynamical variables, defined by the equations 
\begin{equation}\label{oim_sys0}
\begin{aligned}
U_{n+i}=F(U_{n+i-1},V_{n+j-1},U_{n+i-2},V_{n+j-2},\ldots,U_{n},V_{n},u_{n},v_{n},\ldots,u_{n\pm s},v_{n\pm s}),\\
V_{n+j}=G(U_{n+i-1},V_{n+j-1},U_{n+i-2},V_{n+j-2},\ldots,U_{n},V_{n},u_{n},v_{n},\ldots,u_{n\pm s},v_{n\pm s}),
\end{aligned}
\end{equation}
where $\frac{\partial F}{\partial U_n}\neq 0$, $\frac{\partial F}{\partial V_n}\neq 0$, $\frac{\partial G}{\partial U_n}\neq 0$, $\frac{\partial G}{\partial V_n}\neq 0$, is called a generalized invariant manifold for equation (\ref{main2}), if the conditions hold
\begin{equation}\label{cond_sys0}
\begin{aligned}
\left.\frac{d}{dt}F-D_n^i\frac{d U_n}{dt} \right|_{\eqref{sys0},\eqref{lin_sys0},\eqref{oim_sys0}}=0, \\
\left.\frac{d}{dt}G-D_n^j\frac{d V_n}{dt} \right|_{\eqref{sys0},\eqref{lin_sys0},\eqref{oim_sys0}}=0.
\end{aligned}
\end{equation}
\end{definition}

%Nonlinear invariant manifolds and separation of the variables for the 

\section{Generalized invariant manifold}

Here we investigate the following differential-difference equation \cite{Ruijsenaars} 
\begin{align}\label{RTL0}
x_{n,tt}=&(1+hx_{n,t})(1+hx_{n-1,t})\frac{e^{x_{n-1}-x_n}}{1+h^2e^{x_{n-1}-x_n}} \nonumber\\
&-(1+hx_{n,t})(1+hx_{n+1,t})\frac{e^{x_{n}-x_{n+1}}}{1+h^2e^{x_{n}-x_{n+1}}}
\end{align} 
called relativistic Toda lattice or the Ruijsenaars-Toda lattice, here $h=\frac{1}{c}$, where $c$ is the light speed. Under the substitution 
\begin{align}\label{uvx}
u_n=\frac{h\left(1+hx_{n,t}\right)e^{x_{n}-x_{n+1}}}{1+h^2e^{x_{n}-x_{n+1}}}, \qquad v_n=\frac{1+hx_{n,t}}{h(1+h^2e^{x_{n}-x_{n+1}})}
\end{align} 
the lattice \eqref{RTL0} reduces to a system of Toda type 
\begin{equation}
\begin{aligned}\label{RTL}
u_{n,t}&=u_n\left(u_{n-1}-u_{n+1}+v_n-v_{n+1}\right),\\
v_{n,t}&=v_n\left(u_{n-1}-u_{n}\right).
\end{aligned} 
\end{equation}
Below we concentrate on \eqref{RTL} since it is more easy to handle with. Passing from \eqref{RTL} to the original form is discussed at the end of \S 3.3.2. 

In order to look for the generalized invariant manifold for \eqref{RTL} first we find its linearization:
\begin{equation}
\begin{aligned}\label{RTL_lin0}
U_{n,t}&=\left(u_{n-1}-u_{n+1}+v_n-v_{n+1}\right)U_n+u_n\left(U_{n-1}-U_{n+1}+V_n-V_{n+1}\right),\\
V_{n,t}&=\left(u_{n-1}-u_{n}\right)V_n+v_n\left(U_{n-1}-U_{n}\right).
\end{aligned} 
\end{equation}

\begin{example}\label{Ex2}
It can be checked by a direct computation that the following system of equations 
\begin{equation}
\begin{aligned}\label{simpleGIM}
U_{n+1}&=\frac{v_{n,t}u_{n+1}-v_{n}u_{n+1,t}}{v_{n,t}u_{n}-v_{n}u_{n,t}}U_{n}+ \frac{u_{n}u_{n+1,t}-u_{n+1}u_{n,t}}{v_{n,t}u_{n}-v_{n}u_{n,t}}V_{n},\\
V_{n+1}&=\frac{v_{n+1,t}v_{n}-v_{n+1}v_{n,t}}{v_{n}u_{n,t}-v_{n,t}u_{n}}U_{n}+ \frac{v_{n+1}u_{n,t}-u_{n}v_{n+1,t}}{{v_{n}u_{n,t}-v_{n,t}u_{n}}}V_{n}
\end{aligned} 
\end{equation}
defines a generalized invariant manifold for \eqref{RTL} since it is compatible with the linearized equation \eqref{RTL_lin0}.
Let us observe that system \eqref{simpleGIM} admits exact solutions of the form
\begin{equation}
\begin{aligned}\label{simpleGIMsol}
U_{n}&=c_1(u_n+tu_{n,t})+c_2u_{n,t},\\
V_{n}&=c_1(v_n+tv_{n,t})+c_2v_{n,t},
\end{aligned} 
\end{equation}
where $c_1$ and $c_2$ are arbitrary constants. These functions are connected with the classical symmetries $u_{n,\tau}=U_n$, $v_{n,\tau}=V_n$ of the system \eqref{RTL} or the same  
\begin{equation*}
\begin{aligned}
&u_{n,\tau}=c_1 u_n+(c_2+c_1 t)u_{n}(u_{n-1}-u_{n+1}+v_n-v_{n+1}),\\
&v_{n,\tau}=c_1 v_n+(c_2+c_1 t)v_{n}(u_{n-1}-u_{n}).
\end{aligned} 
\end{equation*}
Obviously stationary parts of these symmetries generate a solution to \eqref{RTL}:
\begin{equation*}
\begin{aligned}
&u_n(t)=\frac{n}{\varepsilon+t}+c_3,\\
&v_n(t)=-\frac{n}{\varepsilon+t}+c_4,
\end{aligned} 
\end{equation*}
where $\varepsilon=\frac{c_2}{c_1}>0,$ $c_1\neq 0$.
\end{example}

However the class of the generalized invariant manifolds is not exhausted by the classical and higher symmetries only. Mostly interesting cases contain several constant parameters, moreover they are not solved in a closed form. Below we study such kind example for the relativistic Toda lattice.

At first we simplify the linearized equation by using the substitution 
\begin{align}\label{UVPQ}
U_n=u_n\left(P_{n+1}-P_n\right), \qquad V_n=v_n\left(Q_{n+1}-Q_n\right).
\end{align}
Briefly explain the origin of the substitution. It is found by applying the algorithm of the formal linearization to the substitution
\begin{align}\label{zam}
u_n=e^{p_{n+1}-p_n}, \qquad v_n=e^{q_{n+1}-q_n}
\end{align}
relating the system \eqref{RTL} with its close counterpart
\begin{equation*}
\begin{aligned}\label{RTLpq}
p_{n,t}&=-e^{p_{n}-p_{n-1}}-e^{p_{n+1}-p_{n}}-e^{q_{n+1}-q_{n}},\\
q_{n,t}&=-e^{p_{n}-p_{n-1}}.
\end{aligned} 
\end{equation*}

It is remarkable that the substitution \eqref{UVPQ} found due to the manipulations essentially simplifies the linearized equation. The resulting system of linear equations contains only three functional parameters $u_{n-1}$, $u_n$ and $v_n$  
\begin{equation}
\begin{aligned}\label{RTL_lin1}
&P_{n,t}=u_{n-1}P_{n-1}+(u_n-u_{n-1})P_n-u_nP_{n+1}+v_n(Q_n-Q_{n+1}),\\
&Q_{n,t}=u_{n-1}(P_{n-1}-P_{n}),
\end{aligned} 
\end{equation}
meanwhile system \eqref{RTL_lin0} contains extra variables $u_{n+1}$, $v_{n+1}$. Due to this simplification we can assume that the desired invariant manifold is of the form
\begin{equation}
\begin{aligned}\label{RTL_oim1}
&P_{n+1}=F\left(P_n,Q_n,u_n,v_n,u_{n-1}\right),\\
&Q_{n+1}=G\left(P_n,Q_n,u_n,v_n,u_{n-1}\right),
\end{aligned} 
\end{equation}
i.e. it might depend on the same set of the lowercase variables.

We search the functions $F$ and $G$ using the compatibility condition for the systems of equations \eqref{RTL_lin1} and \eqref{RTL_oim1}. We omit the calculations and give only the answer (see the Appendix for calculations). 

Thus, system of equations \eqref{RTL_oim1} is compatible with the linearized equation \eqref{RTL_lin1}  if and only if it is of the form
\begin{equation}
\begin{aligned}\label{RTL_oim}
&P_{n+1}=\frac{\left(\lambda Q_n+(v_n-\lambda)P_n\right)^2}{\lambda u_nP_n}-\frac{C}{u_nP_n},\\
&Q_{n+1}=-\frac{v_n-\lambda}{\lambda}P_n-Q_n,
\end{aligned} 
\end{equation}
where $\lambda$, $C$ are arbitrary constants. By virtue of relations \eqref{RTL_oim}, equation \eqref{RTL_lin1} can be reduced to the form:
\begin{equation}
\begin{aligned}\label{RTL_lin}
&P_{n,t}=(u_n-u_{n-1}+v_n-\lambda)P_n+2\lambda Q_n,\\
&Q_{n,t}=-u_{n-1}P_n+\lambda \frac{Q^2_n}{P_n}-\frac{C}{P_n}.
\end{aligned} 
\end{equation}
\begin{theorem}\label{ThRTL}
Systems of equations \eqref{RTL_oim} and \eqref{RTL_lin} are compatible if and only if the functions $u_n$ and $v_n$ solve the system  \eqref{RTL}.
\end{theorem}
Theorem \ref{ThRTL} implies that a pair of equations \eqref{RTL_oim} and \eqref{RTL_lin} can be interpreted as a nonlinear Lax pair for the equation \eqref{RTL}.

\section{Constructing of particular solutions via generalized invariant manifolds}

Let us investigate the behavior of the system \eqref{RTL_lin} in the vicinity of its singular point $ \lambda = \infty $. For this purpose, we construct a formal asymptotic expansion of the solution in a series in negative powers of the parameter $\lambda$: 
\begin{equation}\label{PQlambda}
P_n=\sum_{j=0}^\infty \alpha_j\lambda^{-j}, \qquad Q_n=\sum_{j=0}^\infty \beta_j\lambda^{-j},
\end{equation}
where $\alpha_j=\alpha_j(n,t)$ and $\beta_j=\beta_j(n,t)$.
We assume that parameter $C$ is also represented as a formal series
\begin{equation}\label{Clambda}
C(\lambda)=\sum_{j=-1}^\infty c_j\lambda^{-j},
\end{equation}
with constant coefficients $c_j$. By substituting  the series  \eqref{PQlambda} and \eqref{Clambda} into the first equation of system \eqref{RTL_lin} and by comparing the coefficients at the powers of $\lambda$ we obtain a sequence of the equations for determining unknown coefficients $\alpha_j$ and $\beta_j$ 
\begin{equation}
\begin{aligned}\label{coeff1}
&\lambda:  \qquad \alpha_0-2\beta_0=0,\\
&\lambda^0:  \quad \, \, \, \dot{\alpha}_0=(u_n-u_{n-1}+v_n)\alpha_0-\alpha_1+2\beta_1,\\
&\lambda^{-1}: \quad \dot{\alpha}_1=(u_n-u_{n-1}+v_n)\alpha_1-\alpha_2+2\beta_2,\\
&\lambda^{-2}: \quad \dot{\alpha}_2=(u_n-u_{n-1}+v_n)\alpha_2-\alpha_3+2\beta_3,\\
& \ldots\ldots\ldots\ldots.
\end{aligned} 
\end{equation}
Similarly from the second equation in \eqref{RTL_lin} we get
\begin{equation}
\begin{aligned}\label{coeff2}
&\lambda:  \qquad \beta_0^2=-c_{-1},\\
&\lambda^0:  \quad \, \, \, \alpha_0\dot{\beta}_0=-u_{n-1}\alpha_0^2+2\beta_0\beta_1+c_0,\\
&\lambda^{-1}: \quad \alpha_1\dot{\beta}_0+\alpha_0\dot{\beta}_1=-2u_{n-1}\alpha_0\alpha_1+\beta_1^2+2\beta_0\beta_2+c_1,\\
&\lambda^{-2}: \quad \alpha_2\dot{\beta}_0+\alpha_1\dot{\beta}_1+\alpha_0\dot{\beta}_2=-u_{n-1}\left(2\alpha_0\alpha_2+\alpha_1^2\right)+2\left(\beta_1\beta_2+\beta_0\beta_3\right),\\
& \ldots\ldots\ldots\ldots.
\end{aligned} 
\end{equation}
Several first equations from the systems \eqref{coeff1} and \eqref{coeff2} immediately imply
\begin{align*}
&\alpha_0=2\beta_0, \qquad \beta_0=\sqrt{-c_{-1}},\\
&\alpha_1=(u_n-u_{n-1}+v_n)\alpha_0+2\beta_1, \qquad \beta_1=\frac{1}{2\beta_0}(u_{n-1}\alpha_0^2-c_0),\\
&\alpha_2=(u_n-u_{n-1}+v_n)\alpha_1+2\beta_2-\dot{\alpha}_1, \qquad \beta_2=\frac{1}{2\beta_0}(\alpha_0\dot{\beta}_1+2u_{n-1}\alpha_0\alpha_1-\beta_1^2-c_1).
\end{align*}
Obviously, for $j=0,1,2$ the coefficients $\alpha_j$, $\beta_j$ are polynomials in the variables $u_n$, $u_{n-1}$, $v_n$ and  their derivatives with respect to $t$. The next statement is easily proved by induction.
\begin{proposition}
For any $j\geq 0$ coefficients $\alpha_j$ and $\beta_j$ of the formal series \eqref{PQlambda} are polynomials in the variables  $u_n$, $u_{n-1}$, $v_n$ and  their derivatives with respect to $t$:
\begin{align*}
&\alpha_j=f^{(1)}_j\left(u_n,u_{n-1},v_n,\dot{u}_n,\dot{u}_{n-1},\dot{v}_n,\ldots\right),\\
&\beta_j=f^{(2)}_j\left(u_n,u_{n-1},v_n,\dot{u}_n,\dot{u}_{n-1},\dot{v}_n,\ldots\right).
\end{align*}
\end{proposition}
Let us now substitute the formal series \eqref{PQlambda} and \eqref{Clambda} into the system \eqref{RTL_oim} and derive two more sequences of the equations. Namely from the second equation in \eqref{RTL_oim} we obtain
\begin{equation}
\begin{aligned}\label{coeff3}
&\bar{\beta}_0=\alpha_0-\beta_0,\\
&\bar{\beta}_1=\alpha_1-\beta_1-\alpha_0v_n,\\
&\bar{\beta}_2=\alpha_2-\beta_2-\alpha_1v_n,\\
&\ldots\ldots\ldots,\\
&\bar{\beta}_j=\alpha_j-\beta_j-\alpha_{j-1}v_n.
\end{aligned} 
\end{equation}
Here the notation $\bar{\beta}_j=D_n\beta_j$ is accepted where $D_n$ stands for the operator shifting the discrete argument $n$ such that $D_nf(n)=f(n+1)$. The first equation in \eqref{RTL_oim} implies
\begin{equation}
\begin{aligned}\label{coeff4}
&\left(\beta_0-\alpha_0\right)^2+c_{-1}=0,\\
&u_n\bar{\alpha}_0\alpha_0=2(\beta_0-\alpha_0)(\beta_1-\alpha_1)+2v_n(\beta_0-\alpha_0)+c_0,\\
&u_n(\bar{\alpha}_0\alpha_1+\bar{\alpha}_1\alpha_0)=\left(\beta_1-\alpha_1\right)^2+2(\beta_0-\alpha_0)(\beta_2-\alpha_2)\\
&\qquad \qquad \qquad \qquad +2v_n\left((\beta_0-\alpha_0)\alpha_1+\left(\beta_1-\alpha_1\right)\alpha_0\right)+\alpha_0^2v_n^2+c_1,\\
& \ldots\ldots\ldots\ldots
\end{aligned} 
\end{equation}
From \eqref{coeff3} and \eqref{coeff4} one can easily observe that the coefficients $\alpha_0$ and $\beta_0$ do not depend on  $n$. Further analysis shows that
\begin{equation*}
\beta_1=\alpha_0u_{n-1}+(1-\alpha_0)v_{n-1}-\frac{c_0}{\alpha_0}
\end{equation*}
and
\begin{equation*}
\alpha_1=\alpha_0v_{n}+\alpha_0(u_{n-1}+u_n)+(1-\alpha_0)(v_n+v_{n-1}).
\end{equation*}
It is easy to prove that due to the relations \eqref{coeff3} and \eqref{coeff4} all the coefficients $ \alpha_j $, $ \beta_j $ can be  expressed in terms of the variables $ u_n $, $ v_n $ and their shifts.
In other words, the following statement is true
\begin{proposition}
The coefficients $\alpha_j$ and $\beta_j$ of the series \eqref{PQlambda} are polynomials in the variables  $u_n$, $v_n$ and their shifts:
\begin{align*}
&\alpha_j=g^{(1)}_j\left(u_n,v_n,u_{n\pm 1},v_{n\pm 1},\ldots\right),\\
&\beta_j=g^{(2)}_j\left(u_n,v_n,u_{n\pm 1},v_{n\pm 1},\ldots\right).
\end{align*}
\end{proposition}
The following important fact follows from Propositions 4.1 and 4.2.
\begin{proposition}
If the formal series \eqref{PQlambda} are terminated at $j=N$, i.e. the conditions hold $\alpha_j\equiv 0$, $\beta_j\equiv 0$ for $j\geq N$, then the corresponding solution $u(n,t)$, $v(n,t)$ of the lattice \eqref{RTL} solves the system of the ordinary differential equations:
\begin{equation}
\begin{aligned}\label{fN}
&f^{(1)}_N\left(u_n,u_{n-1},v_n,\dot{u}_n,\dot{u}_{n-1},\dot{v}_n,\ldots\right)=0,\\
&f^{(2)}_N\left(u_n,u_{n-1},v_n,\dot{u}_n,\dot{u}_{n-1},\dot{v}_n,\ldots\right)=0
\end{aligned}
\end{equation}
and simultaneously the system of the ordinary difference equations:
\begin{equation}
\begin{aligned}\label{gN}
&g^{(1)}_N\left(u_n,v_n,u_{n\pm 1},v_{n\pm 1},\ldots\right)=0,\\
&g^{(2)}_N\left(u_n,v_n,u_{n\pm 1},v_{n\pm 1},\ldots\right)=0.
\end{aligned}
\end{equation}
\end{proposition}
Note that the equalities \eqref{fN} and \eqref{gN} define some invariant manifolds for the equation \eqref{RTL}. 

\section{Derivation of the continuous and discrete Dubrovin type equations}

The system of equations \eqref{fN}, \eqref{gN} is written in an absolutely implicit form. It is convenient to rewrite it using the following technique. Since now series \eqref{PQlambda}, \eqref{Clambda} are polynomials in $\lambda$, it is possible to use as functional parameters not the coefficients $\alpha_i$ and $\beta_i$, but the roots of these polynomials (see, for instance, \cite{DubrovinMatveevNovikov}).

Let us change the variables in the systems \eqref{RTL_oim} and \eqref{RTL_lin} by setting \linebreak $\lambda Q_n=R_n$, $\lambda C=c$ as a result we obtain
\begin{equation}
\begin{aligned}\label{RTL_oim2}
&\lambda u_nP_nP_{n+1}=\left(R_n+(v_n-\lambda)P_n\right)^2-c,\\
&R_{n+1}=-(v_n-\lambda)P_n-R_n,
\end{aligned} 
\end{equation}
\begin{equation}
\begin{aligned}\label{RTL_lin2}
&\dot{P}_n=(u_n-u_{n-1}+v_n-\lambda)P_n+2 R_n,\\
&P_n\dot{R}_n=-\lambda u_{n-1}P^2_n+R^2_n-c.
\end{aligned} 
\end{equation}
We assume that system \eqref{RTL_oim2}, \eqref{RTL_lin2} has a solution polynomially depending on the parameter~$\lambda$: 
\begin{equation}\label{PQ_polinom}
P_n(\lambda)=2\prod_{i=1}^N(\lambda-\gamma_n^{(i)}), \qquad R_n(\lambda)=\prod_{i=1}^{N+1}(\lambda-\rho_n^{(i)}),
\end{equation}
where $\gamma_n^{(i)}=\gamma^{(i)}(n,t)$, $\rho_n^{(i)}=\rho^{(i)}(n,t)$.

As it was established earlier, in this case the solution of the nonlinear lattice \eqref{RTL} satisfies some ordinary differential and ordinary discrete equations, which are expressed in a complicated way in terms of the functions $u(n,t)$, $v(n,t)$ and their derivatives and shifts. It is convenient to rewrite these equations in terms of the roots $\gamma_n^{(j)}$ and $\rho_n^{(j)}$ of the polynomials \eqref{PQ_polinom}. In this case, function $c=c(\lambda)$ should also be a polynomial. It is convenient to set it in the form of a product
\begin{equation}\label{C_polinom}
c(\lambda)=\prod_{i=1}^{2N+2}(\lambda-\lambda_i)=\nu^2(\lambda).
\end{equation}
Note that formula \eqref{C_polinom} defines a hyperelliptic curve of some algebro-geometric solution of the equation \eqref{RTL}.

Let us substitute products \eqref{PQ_polinom} and \eqref{C_polinom} into equations \eqref{RTL_oim2}, \eqref{RTL_lin2}  and get:
\begin{equation}
\begin{aligned}\label{RTL_Oimpolinom}
&\prod_{i=1}^{N+1}(\lambda-\rho_n^{(i)})+2(v_n-\lambda)\prod_{i=1}^N(\lambda-\gamma_n^{(i)})=\varepsilon\sqrt{4\lambda u_n\prod_{i=1}^N(\lambda-\gamma_n^{(i)})(\lambda-\gamma_{n+1}^{(i)})+\nu^2(\lambda)},\\
&\prod_{i=1}^{N+1}(\lambda-\rho_{n+1}^{(i)})+2(v_n-\lambda)\prod_{i=1}^N(\lambda-\gamma_n^{(i)})+\prod_{i=1}^{N+1}(\lambda-\rho_n^{(i)})=0,
\end{aligned} 
\end{equation}
\begin{equation}
\begin{aligned}\label{RTL_Linpolinom}
&-\sum_{j=1}^N \dot{\gamma}_n^{(j)} \prod_{i\neq j}^N(\lambda-\gamma_n^{(i)})=(u_n-u_{n-1}+v_n-\lambda)\prod_{i=1}^N(\lambda-\gamma_n^{(i)})+\prod_{i=1}^{N+1}(\lambda-\rho_n^{(i)}),\\
&\prod_{i=1}^N(\lambda-\gamma_n^{(i)})\left[\sum_{j=1}^{N+1}\dot{\rho}_n^{(j)}\prod_{i\neq j}^{N+1}(\lambda-\rho_n^{(i)})\right]=2\lambda u_{n-1}\prod_{i=1}^N(\lambda-\gamma_n^{(i)})^2 \\
& \qquad \qquad \qquad \qquad \qquad \qquad \qquad \qquad \quad -\frac{1}{2}\prod_{i=1}^{N+1}(\lambda-\rho_{n+1}^{(i)})^2+\frac{1}{2}\nu^2(\lambda).
\end{aligned} 
\end{equation}
Here for the simplicity we use notations $\gamma_{n+1}^{(i)}=\gamma^{(i)}(n+1,t)$, $\rho_{n+1}^{(i)}=\rho^{(i)}(n+1,t)$, \linebreak $\dot{\gamma}_n^{(j)}=\frac{d}{dt}\gamma^{(j)}(n,t)$, $\dot{\rho}_n^{(j)}=\frac{d}{dt}\rho^{(j)}(n,t)$, parameter $\varepsilon$ takes any of two values $\pm 1$. 

Equating the coefficients at the two highest powers of $\lambda$, we find the relations
\begin{align}
&u_n=-\frac{1}{2}\sum_{i=1}^N\rho_{n+1}^{(i)}+\frac{1}{4}\sum_{i=1}^{2N+2}\lambda_i, \label{un}\\
&v_n=\frac{1}{2}\sum_{i=1}^N\left(\rho_{n+1}^{(i)}+\rho_n^{(i)}\right)-\sum_{i=1}^{N}\gamma_n^{(i)}.\label{vn}
\end{align} 
Afterwards by taking $\lambda=0$ in \eqref{RTL_Oimpolinom} we get two relations
\begin{align}
(-1)^{N+1}\prod_{i=1}^{N+1}\rho_n^{(i)}=-\varepsilon \nu(0), \qquad \quad
v_n=(-1)^N\frac{\varepsilon \nu(0)}{\prod_{i=1}^N\gamma_n^{(i)}}.\label{vn2}
\end{align} 
Then due to \eqref{vn} and \eqref{vn2} we obtain:
\begin{align}\label{sumro}
\sum_{i=1}^{N+1}(\rho_{n+1}^{(i)}+\rho_n^{(i)})=2\sum_{i=1}^{N}\gamma_n^{(i)}+2(-1)^N\frac{\varepsilon \nu(0)}{\prod_{i=1}^N\gamma_n^{(i)}}.
\end{align} 
Let's set now $\lambda=\gamma_n^{(j)}$ in \eqref{RTL_Oimpolinom} and deduce equations relating roots $\gamma_n^{(j)}$, $\rho_n^{(j)}$ and $\rho_{n+1}^{(j)}$:
\begin{equation}
\begin{aligned}\label{RTL_oim_lam_gj}
&\prod_{i=1}^{N+1}(\gamma_n^{(j)}-\rho_n^{(i)})=\varepsilon\nu(\gamma_n^{(j)}),\\
&\prod_{i=1}^{N+1}(\gamma_n^{(j)}-\rho_{n+1}^{(i)})+\prod_{i=1}^{N+1}(\gamma_n^{(j)}-\rho_n^{(i)})=0.
\end{aligned} 
\end{equation}
From the first equation \eqref{RTL_Linpolinom} by imposing $\lambda=\gamma_n^{(j)}$ and transforming due to the first equation in  \eqref{RTL_oim_lam_gj} we derive a system of the first order  differential equations for the variables $\gamma_n^{(j)}$:
\begin{align}\label{dotgj}
\dot{\gamma}_n^{(j)} =-\frac{\varepsilon\nu(\gamma_n^{(j)})}{\prod_{i\neq j}^N(\gamma_n^{(j)}-\gamma_n^{(i)})}, \quad j=1,2,\ldots,N,
\end{align} 
which are nothing else but the usual Dubrovin equations.

From equations \eqref{sumro} and \eqref{RTL_oim_lam_gj} one can derive the following system of discrete equations for $\gamma_n^{(j)}$:
\begin{align}
-\sum_{j=1}^N\left(\frac{\nu(\gamma_n^{(j)})}{\gamma_n^{(j)}(\gamma_n^{(j)}-\gamma_{n+1}^{(s)})\prod_{k\neq j}^{N}(\gamma_n^{(j)}-\gamma_n^{(k)})}\right)&+\frac{\nu(\gamma_{n+1}^{(s)})}{\gamma_{n+1}^{(s)}\prod_{k=1}^{N}(\gamma_{n+1}^{(s)}-\gamma_n^{(k)})} \nonumber\\
&+(-1)^N\frac{\nu(0)}{\gamma_{n+1}^{(s)}\prod_{i=1}^N \gamma_n^{(i)}}=\varepsilon, \label{gjnp1}
\end{align} 
where $s=1,2,\ldots,N$. Note that the derivation of these equations is not trivial; it requires ingenuity and some effort.

System of $N$ first order difference equations  \eqref{gjnp1} can be interpreted as a discrete version of the Dubrovin equations for the lattice \eqref{RTL}. 

Finally, due to the equations \eqref{sumro} and \eqref{RTL_oim_lam_gj}, we find an expression for the sought function  $u_n$ (see above \eqref{un}) in terms of the variables $\gamma_n^{(j)}$: 
\begin{align}\label{un2}
u_n=-\frac{1}{2}\left(\sum_{i=1}^N\gamma_n^{(i)}+\sum_{i=1}^N\frac{\varepsilon \nu(\gamma_n^{(i)})}{\gamma_n^{(i)} \prod_{k\neq i}^N(\gamma_n^{(i)}-\gamma_n^{(k)})}+(-1)^N\frac{\varepsilon \nu(0)}{\prod_{i=1}^N \gamma_n^{(i)}}\right)+\frac{1}{4}\sum_{i=1}^{2N+2}\lambda_i.
\end{align}

System of equations \eqref{dotgj} admits integrals of the form (cf. \cite{DubrovinMatveevNovikov}):
\begin{equation}
\begin{aligned}
& \sum^N_{j=1}\int_{\gamma^{(j)}(n,0)}^{\gamma^{(j)}(n,t)}\gamma^k\frac{d \gamma}{\nu(\gamma)}=0, \qquad k=0,1,2,\ldots,N-2,\\
& \sum^N_{j=1}\int_{\gamma^{(j)}(n,0)}^{\gamma^{(j)}(n,t)}\gamma^{N-1}\frac{d \gamma}{\nu(\gamma)}=-\varepsilon t,
\end{aligned}
\end{equation}
that are derived directly from the system by simple transformations and integration. Dependance of the functions $\gamma^{(j)}(n,t)$ on the variables $n$ is due to the discrete system \eqref{gjnp1}.

\section{Particular solutions}

Let us assume that $N=1$ in the formulas \eqref{PQ_polinom}, \eqref{C_polinom}. Then solution $(u_n,v_n)$ to the lattice \eqref{RTL}
is found as follows
\begin{align}
& u_n=-\frac{1}{2}\left(\gamma_n+\frac{\varepsilon \nu(\gamma_n)}{\gamma_n}-\frac{\varepsilon\nu(0)}{\gamma_n}\right)+\frac{1}{4}\sum^{4}_{i=1}\lambda_i, \label{unN1}\\
& v_n=-\frac{\varepsilon\nu(0)}{\gamma_n}, \label{vnN1}
\end{align}
where $\gamma_n$ solves the corresponding continuous and discrete Dubrovin type equations 
\begin{equation}
\begin{aligned}\label{dotgjN1}
& \dot{\gamma}_n=-\varepsilon\nu(\gamma_n), \\
& -\frac{\nu(\gamma_n)}{\gamma_n(\gamma_n-\gamma_{n+1})}+\frac{\nu(\gamma_{n+1})}{\gamma_{n+1}(\gamma_{n+1}-\gamma_n)}-\frac{\nu(0)}{\gamma_n\gamma_{n+1}}=\varepsilon.
\end{aligned} 
\end{equation}
Here $\nu^2(\gamma_n)=\prod_{i=1}^{4}(\gamma_n-\lambda_i)$.
%\begin{align*}
%\nu^2(\gamma_n)=\prod_{i=1}^{4}(\gamma_n-\lambda_i). % \qquad \nu^2(\gamma_{n+1})=\prod_{i=1}^{4}(\gamma_{n+1}-\lambda_i), \qquad \nu^2(0)=\prod_{i=1}^{4}(\lambda_i).
%\end{align*} 

The next statement is proved by a direct computation.
\begin{proposition} 
The overdetermined system of the equations \eqref{dotgjN1} is compatible.
\end{proposition}

\subsection{Construction of soliton-like solutions}

In order to construct solution $(u_n,v_n)$ we first have to solve equations \eqref{dotgjN1}. For the simplicity we assume that $\lambda_4=\lambda_2$, $\lambda_3=\lambda_1$, $\varepsilon=1$ and then easily get:
\begin{align}
& \dot{\gamma}_n=-(\gamma_n-\lambda_2)(\gamma_n-\lambda_1), \label{eq1_soliton}\\
& \frac{(\gamma_{n+1}-\lambda_1)(\gamma_{n+1}-\lambda_2)}{\gamma_{n+1}(\gamma_{n+1}-\gamma_{n})}-\frac{(\gamma_n-\lambda_1)(\gamma_n-\lambda_2)}{\gamma_n(\gamma_n-\gamma_{n+1})}-\frac{\lambda_1\lambda_2}{\gamma_{n}\gamma_{n+1}}-1=0. \label{eq2_soliton}
\end{align} 
By integrating equation \eqref{eq1_soliton} we obtain:
\begin{align}\label{soliton_gamn}
\gamma_n(t)=\frac{\lambda_1 e^{(\lambda_1-\lambda_2)(t+C(n))}-\lambda_2}{e^{(\lambda_1-\lambda_2)(t+C(n))}-1}.
\end{align} 
It remains to determine $C(n)$. To this end we first simplify \eqref{eq2_soliton} and find:
\begin{equation}\label{simple}
\gamma_{n+1}(t)=\frac{\lambda_1\lambda_2}{\lambda_1+\lambda_2-\gamma_n(t)}.
\end{equation}
Afterward we substitute \eqref{soliton_gamn} into \eqref{simple} and by means of elementary transformations we obtain the following equation for $C(n)$
\begin{align*}
\lambda_2e^{(\lambda_2-\lambda_1)C(n)}-\lambda_1e^{(\lambda_2-\lambda_1)C(n+1)}=0,
\end{align*} 
that is easily solved
\begin{align*}
C(n)=\frac{1}{\lambda_2-\lambda_1}\left(n\ln\frac{\lambda_1}{\lambda_2}-i\pi-c\right).
\end{align*} 
Thus, the sought function $\gamma_n(t)$ takes the form 
\begin{align*}
\gamma_n(t)=\lambda_2\frac{e^{(n-1)\theta+wt+c}+1}{e^{n\theta+wt+c}+1},
\end{align*} 
where $\theta=\ln\frac{\lambda_2}{\lambda_1}$, $w=\lambda_1-\lambda_2$, $c$ is a constant. Now by substituting the found representation of the function $\gamma_n(t)$ into the expressions \eqref{unN1} and \eqref{vnN1}, we obtain the desired solution to \eqref{RTL}
\begin{align*}
&u_n(t)=\lambda_1\frac{e^{(n+1)\theta+wt+c}+1}{e^{n\theta+wt+c}+1},\\
&v_n(t)=-\lambda_1\frac{e^{n\theta+wt+c}+1}{e^{(n-1)\theta+wt+c}+1}.
\end{align*} 
These formulae can be written in a compact form
\begin{align*}
&u_n(t)=\sqrt{\lambda_1 \lambda_2 }\frac{\cosh (\hat{w}+\frac{1}{2}\theta)}{\cosh \hat{w}},\quad \hat{w}=\frac{1}{2}(n\theta+wt+c),\\
&v_n(t)=-\sqrt{\lambda_1 \lambda_2 }\frac{\cosh \hat{w}}{\cosh (\hat{w}-\frac{1}{2}\theta)}.
\end{align*} 
The found solution has the form of a kink. Pictures 1, 2 show the graphs of the functions $u_n(t)$ and $v_n(t)$ corresponding to the following choice of parameter values: $\lambda_1=1$, $\lambda_2=4$, $c=0$, $t=0$.

\vspace{3mm}

\begin{tabular}{ l l }
\includegraphics[scale=0.6]{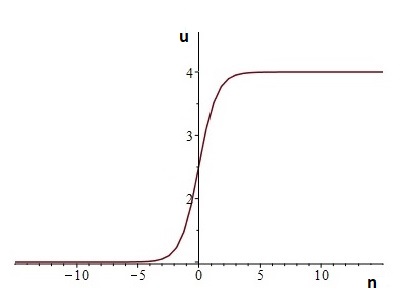} & \hspace{0.7 cm} \includegraphics[scale=0.6]{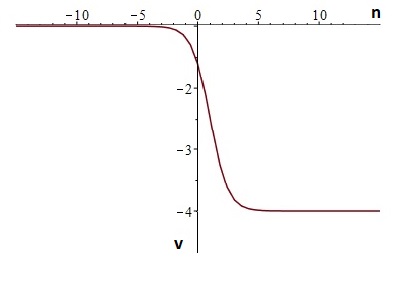} \\
\hspace{1.5 cm}  Pic. 1. Graph of $u_n$. & \hspace{2.5 cm}  Pic. 2. Graph of $v_n$.
\end{tabular}
 
\subsection{A periodic solution}

In this section we concentrate on the case when roots satisfy an involution $\lambda_3=-\lambda_1$,  $\lambda_4=-\lambda_2$. The time evolution is governed by the equation 
\begin{align}\label{gammat_ps}
\gamma_{n,t}=\sqrt{(\gamma_n^2-\lambda_1^2)(\gamma_n^2-\lambda_2^2)}.
\end{align} 
The discrete Dubrovin equation \eqref{dotgjN1}, defining a B$\ddot{a}$cklund transformation for \eqref{gammat_ps} now degenerates into a point transformation of the form
\begin{align}\label{gammap1_ps}
\gamma_{n+1}=-\frac{\lambda_1\lambda_2}{\gamma_n}.
\end{align} 
This circumstance implies that $\gamma_n(t)$ is periodic on $n$ with the period $2$. Equation \eqref{gammap1_ps} is linearized by a point transformation $\theta_n=\log \gamma_n$:
\begin{align*}
\theta_{n+1}=-\theta_n+\log (-\lambda_1\lambda_2).
\end{align*} 
Its general solution is $\theta_n=(-1)^na(t)+\frac{1}{2}\log (-\lambda_1\lambda_2)$, it allows to specify dependence of $\gamma_n(t)$ on $n$:
\begin{align}\label{gamma1_ps}
\gamma_n(t)=e^{(-1)^na(t)}\sqrt{-\lambda_1\lambda_2}.
\end{align} 
We change the variables in \eqref{gammat_ps} due to formulas $\gamma_n=\lambda_1y_n$ and $t=\tau/\lambda_2$ and get an integral
\begin{align}\label{gamma_int_ps}
\frac{dy_n}{dt}=\sqrt{(1-y_n^2)(1-k^2y_n^2)}, \qquad k=\frac{\lambda_1}{\lambda_2}
\end{align} 
defining the well known elliptic sinus of Jacobi with the modulus $k$, therefore we have 
\begin{align*}
y_n=\sn(\tau+c(n)).
\end{align*}
Now we turn back to the original variables and arrive at 
\begin{align}\label{gamman_ps}
\gamma_n(t)=\lambda_1\sn(\lambda_2t+c(n)).
\end{align}
Let's specify $c(n)$ due to equation \eqref{gamma1_ps} for $t=0$:
\begin{align*}
\gamma_n(0)=\lambda_1\sn(c(n))=e^{(-1)^na(0)}\sqrt{-\lambda_1\lambda_2}.
\end{align*}
Thus we have
\begin{align}\label{snc_ps}
\sn(c(2n))=\frac{e^{a(0)}}{\sqrt{-k}}, \quad \sn(c(2n+1))=\frac{e^{-a(0)}}{\sqrt{-k}}.
\end{align}
Here $a(0)$ is considered as an arbitrary constant. We assume that parameter $k$ is negative otherwise function $c(n)$ will not be real valued function. We note that function $c(n)$ is periodic with the period $2$. The values of $c(2n):=c_0$ and $c(2n+1):=c_1$ are chosen in such a way that quantities $\cn(c_0)\dn(c_0)$ and $\cn(c_1)\dn(c_1)$ satisfy the relation $\cn(c_0)\dn(c_0)+\sn^2(c_0)\cn(c_1)\dn(c_1)=0$. For this, it is enough that $c_0$ and $c_1$ are located on opposite sides of the maximum point (see Pic. 3).

\begin{tabular}{ l }
\hspace{1 cm} \includegraphics[scale=0.7]{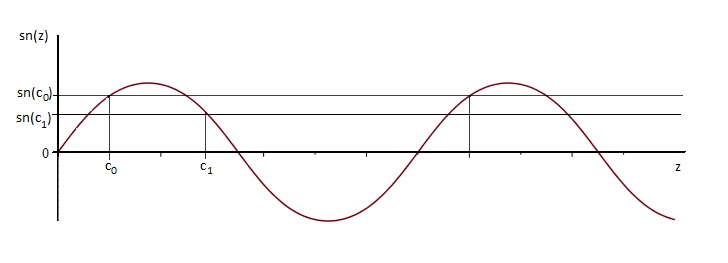} \\
\hspace{6 cm} Pic. 3. Choice of $c_0$, $c_1$.
\end{tabular}

\vspace{1 cm}

In this case the solution is given by
\begin{align}
&u_n(t)=\frac{\lambda_2}{2}\left(|k| \sn(\varkappa(n,t))+\frac{\cn(\varkappa(n,t))\dn(\varkappa(n,t))+1}{\sn(\varkappa(n,t))}\right), \label{un_ps} \\
&v_n(t)=-\frac{\lambda_2}{\sn(\varkappa(n,t))}, \qquad \varkappa(n,t)=\lambda_2t+c(n).  \label{vn_ps}
\end{align}
Note that the functions \eqref{un_ps}, \eqref{vn_ps} have poles at the zeros of the function $\sn(\varkappa(n,t))$. 

One can pass to the variables $x=x_n(t)$ by using the inverse transformation of \eqref{uvx}:
\begin{align}
&x_n-x_{n+1}=\log\left(\frac{u_n}{v_n}\right)+\log\left(\frac{1}{h^2}\right),\label{N1} \\
&\frac{dx_n}{dt}=u_n+v_n+\frac{1}{h}. \label{N2} 
\end{align}
Due to the periodicity property $u_{n+2}=u_n$, $v_{n+2}=v_n$ an explicit expression for $x_n$ is easily found
\begin{align}
&x_{2n}=x_0-n\left(\log\left(\frac{u_0u_1}{v_0v_1}\right)+4\log\left(\frac{1}{h}\right)\right), \quad n\in \Z \\
&x_{2n+1}=x_{2n}-\left(\log\left(\frac{u_0}{v_0}\right)+2\log\left(\frac{1}{h}\right)\right).
\end{align}
Desired $x_0$ is found by integrating \eqref{N2}:
\begin{align*}
x_0=-\frac{t}{h}+\frac{1}{2}\log\left(\sqrt{\frac{1-|k|}{1+|k|}}\right)+\frac{1}{2}\log\left(1-|k|\sn^2(\varkappa(0,t))-\cn(\varkappa(0,t))\dn(\varkappa(0,t))\right).
\end{align*}
Solution \eqref{un_ps}, \eqref{vn_ps} in the limit for $k^2\rightarrow 1$ takes the form 
\begin{align*}
u_n(t)=\lambda_2\coth(\lambda_2t+c(n)), \qquad v_n(t)=-\lambda_2\coth(\lambda_2t+c(n)),
\end{align*}
where $\tanh(c(n))=e^{(-1)^na(0)}$.

\subsection{Construction of elliptic solution}

Let us construct an elliptic solution to system \eqref{dotgjN1}. For simplicity, we introduce notations 
$$e_1=-\lambda_1-\lambda_2-\lambda_3-\lambda_4, \qquad e_2=\lambda_1\lambda_2+\lambda_1\lambda_3+\lambda_1\lambda_4+\lambda_2\lambda_3+\lambda_2\lambda_4+\lambda_3\lambda_4,$$ 
$$e_3=-\lambda_1\lambda_2\lambda_3-\lambda_1\lambda_2\lambda_4-\lambda_1\lambda_3\lambda_4-\lambda_2\lambda_3\lambda_4,  \qquad e_4^2=\lambda_1\lambda_2\lambda_3\lambda_4,$$ set $\varepsilon=-1$, and rewrite system \eqref{dotgjN1} in the form 
\begin{align}
\gamma_{n,t}=&\sqrt{\gamma_n^4+e_1\gamma_n^3+e_2\gamma_n^2+e_3\gamma_n+e_4^2},\label{gt_ell}\\
\gamma_{n+1}=&\frac{2(e_3-e_1e_4)\gamma_n^2+(8e_4^2+e_1e_3-4e_2e_4)\gamma_n+2e_4(e_1e_4-e_3)}{\gamma_n(e_1^2\gamma_n-4e_2\gamma_n+8e_4\gamma_n+4e_1e_4-4e_3)} \nonumber\\
& + \frac{2(e_1e_4-e_3)\sqrt{\gamma_n^4+e_1\gamma_n^3+e_2\gamma_n^2+e_3\gamma_n+e_4^2}}{\gamma_n(e_1^2\gamma_n-4e_2\gamma_n+8e_4\gamma_n+4e_1e_4-4e_3)}. \label{gp1_ell}
\end{align}

We transform the polynomial
\begin{align}\label{G}
G(\gamma_n):=\gamma_n^4+e_1\gamma_n^3+e_2\gamma_n^2+e_3\gamma_n+e_4^2=y_n^2
\end{align}
into the Weierstrass normal form using fractional rational transformation 
\begin{align}\label{transf}
\gamma_n=\frac{2\alpha_0\xi_n-\alpha_0A_2-2A_1}{2\xi_n-A_2}, \qquad y_n=\frac{A_1\eta_n}{\left(\xi_n-\frac{1}{2}A_2\right)^2},
\end{align}
where $\alpha_0$ is a root of the polynomial $G(\gamma_n)$, $A_1=-\frac{1}{4}(4\alpha_0^3+3e_1\alpha_0^2+2e_2\alpha_0+e_3)$, \linebreak $A_2=\frac{1}{6}(6\alpha_0^2+3e_1\alpha_0+e_2)$.
The elliptic curve in the new variables $\xi_n$, $\eta_n$ will have the form
\begin{align*}
4\xi_n^3-g_2\xi_n-g_3=\eta_n^2,
\end{align*}
where $g_2=e_4^2+\frac{1}{12}e_2^2-\frac{1}{4}e_1e_3$, $g_3=-\frac{1}{16}e_3^2+\frac{1}{48}e_1e_2e_3-\frac{1}{216}e_2^3-\frac{1}{16}e_1^2e_4^2+\frac{1}{6}e_2e_4^2$.

The system of equations \eqref{gt_ell}, \eqref{gp1_ell}, due to the transformation \eqref{transf}, goes over to the system
\begin{align}
\xi_{n,t}=&\sqrt{4\xi_n^3-g_2\xi_n-g_3},\label{xit_ell}\\
\xi_{n+1}=&\frac{e_2-6e_4}{12}+\frac{3}{4}\frac{8e_4^2-4e_2e_4+e_1e_3}{12\xi_n+6e_4-e_2}\nonumber\\
&+\frac{9}{2}\frac{(e_1e_4-e_3)^2}{(12\xi_n+6e_4-e_2)^2}-\frac{18(e_1e_4-e_3)\sqrt{4\xi_n^3-g_2\xi_n-g_3}}{(12\xi_n+6e_4-e_2)^2}.\label{xip1_ell}
\end{align}
Equations \eqref{xit_ell}, \eqref{xip1_ell} are consistent, since the point transformation preserves the compatibility condition. Let us construct a general solution to equation \eqref{xit_ell}, which is expressed in terms of the Weierstrass $\wp$-function 
\begin{align}\label{xi_wp}
\xi_n(t)=\wp(t+c(n)).
\end{align}
Next, we need to find the explicit form of function $c(n)$ in formula \eqref{xi_wp} so that it is also a solution to equation \eqref{xip1_ell}.
\begin{proposition}
Function $\xi_n(t)$ can be taken as 
\begin{align}\label{xi_wp_vt}
\xi_n(t)=\wp(t+vn+\beta).
\end{align}
\end{proposition}
\begin{proof} 
Let us rewrite formula \eqref{xip1_ell} in the form
\begin{align}\label{xip1-th}
\xi_{n+1}=\frac{1}{4}\left[\frac{\sqrt{4\xi_n^3-g_2\xi_n-g_3}-\frac{e_1e_4-e_3}{4}}{\xi_n-\frac{e_2-6e_4}{12}}\right]^2-\xi_n-\frac{e_2-6e_4}{12}.
\end{align}
Let us show that equation \eqref{xip1-th} with an appropriate choice of parameters completely coincides with the addition theorem for the Weierstrass $\wp$-function 
\begin{align}\label{add-th}
\wp(u+v)=\frac{1}{4}\left[\frac{\wp'(u)-\wp'(v)}{\wp(u)-\wp(v)}\right]^2-\wp(u)-\wp(v),
\end{align}
where $\wp'(u)$  is the derivative of the function $\wp(u)$. We set $c(n)=vn+\beta$, $u=t+vn+\beta$ in equation \eqref{add-th}, where $v$ is found from equation 
\begin{align*}
\wp(v)=\frac{e_2-6e_4}{12}.
\end{align*}
It is easily verified that $\wp'(v)$ can be chosen as follows
\begin{align*}
\wp'(v)=\sqrt{4\left(\frac{e_2-6e_4}{12}\right)^3-g_2\left(\frac{e_2-6e_4}{12}\right)-g_3}=\frac{e_1e_4-e_3}{4}.
\end{align*}
Then we actually find that equations \eqref{xip1-th} and \eqref{add-th} coincide and function \eqref{xi_wp_vt} is a solution to equations \eqref{xit_ell} and \eqref{xip1_ell}.
\end{proof}
Let us express $\gamma_n(t)$ in terms of $\xi_n(t)$ and write down the solution $u_n(t)$ and $v_n(t)$ of system \eqref{RTL}.  
Using the M$\ddot{o}$bius transform \eqref{transf}, we write a formula for $\gamma_n(t)$ in terms of the known function $\xi_n(t)$:
\begin{align*}
\gamma_n=\frac{2\alpha_0\wp(z)-2A_1-\alpha_0A_2}{2\wp(z)-A_2},
\end{align*}
where $z=t+vn+\beta$. Then the solution to system \eqref{RTL} given by the equalities \eqref{unN1} and \eqref{vnN1} has the form
\begin{align*}
u_n(t)=&-\frac{1}{2}\frac{\alpha_0\wp(z)+r_1}{\wp(z)+r_2}+\frac{1}{2}\frac{A_1\wp'(z)}{(\wp(z)+r_2)(\alpha_0\wp(z)+r_1)}-\frac{1}{2}\frac{e_4(\wp(z)+r_2)}{\alpha_0\wp(z)+r_1}-\frac{1}{4}e_1 ,\\
v_n(t)=&\frac{e_4(\wp(z)+r_2)}{\alpha_0\wp(z)+r_1},
\end{align*}
where $r_1=-(2A_1+\alpha_0A_2)/2$, $r_2=-A_2/2$.

\begin{remark}
In generic case discrete Dubrovin type equation \eqref{gp1_ell} is rather complex. It realizes a B$\ddot{a}$cklund transformation for the equation \eqref{gt_ell}, however it is not reduced to a point transformation. It is worth mentioning that it is closely related to the addition theorem for the Weierstrass $\wp$-function.
\end{remark}

\section*{Conclusion}

The article discusses the problem of constructing algebro-geometric solutions of nonlinear integrable lattices. The novelty of our approach lies in the fact that instead of the usual Lax pair, we use a generalized invariant manifold, which is derived directly from the equation under consideration without invoking additional external information. For the Ruijsenaars-Toda lattice, the GIM derivation is presented in the Appendix. The advantage of our approach is that both continuous and discrete equations of Dubrovin type can be derived from the linearized equation and an appropriately chosen generalized invariant manifold. Some particular solutions for the lattice are presented.

\section*{Acknowledgments}

The work of A.R. Khakimova is supported in part by Young Russian Mathematics award,  the work of A.O.Smirnov was supported by the Ministry of Science and Higher Education of the Russian Federation (grant agreement № FSRF-2020-0004).

\section*{Appendix}

In this section, we describe the calculation process the generalized invariant manifold \eqref{RTL_oim}. 
We will look for a GIM in the form \eqref{RTL_oim1} from the  compatibility condition with system \eqref{RTL_lin1}.
In other words, the functions $F$ and $G$ must satisfy conditions 
\begin{equation}
\begin{aligned}\label{comp_cond_plus_0}
&\frac{d}{dt}F\left(P_n,Q_n,u_n,v_n,u_{n-1}\right) \\
&\qquad \quad-D_n\left(u_{n-1}P_{n-1}+(u_n-u_{n-1})P_n-u_nP_{n+1}+v_n(Q_n-Q_{n+1})\right)=0,\\
&\frac{d}{dt}G\left(P_n,Q_n,u_n,v_n,u_{n-1}\right)-D_n\left(u_{n-1}(P_{n-1}-P_{n})\right)=0
\end{aligned} 
\end{equation}
by virtue of equations \eqref{RTL}, \eqref{RTL_lin1}, \eqref{RTL_oim1} and equalities 
\begin{equation}
\begin{aligned}\label{RTL_oim_m}
&P_{n-1}=f\left(P_n,Q_n,u_{n-1},v_{n-1},u_{n-2}\right),\\
&Q_{n-1}=g\left(P_n,Q_n,u_{n-1},v_{n-1},u_{n-2}\right),
\end{aligned} 
\end{equation}
which in a sense are inverse to \eqref{RTL_oim1}. We note that equations \eqref{RTL_oim_m} also define GIM, so they must satisfy compatibility conditions
\begin{equation}
\begin{aligned}\label{comp_cond_minus_0}
&\frac{d}{dt}f\left(P_n,Q_n,u_{n-1},v_{n-1},u_{n-2}\right)\\
&\qquad \quad -D^{-1}_n\left(u_{n-1}P_{n-1}+(u_n-u_{n-1})P_n-u_nP_{n+1}+v_n(Q_n-Q_{n+1})\right)=0,\\
&\frac{d}{dt}g\left(P_n,Q_n,u_{n-1},v_{n-1},u_{n-2}\right)-D^{-1}_n\left(u_{n-1}(P_{n-1}-P_{n})\right)=0
\end{aligned} 
\end{equation}
due to the same equations \eqref{RTL}, \eqref{RTL_lin1}, \eqref{RTL_oim1} and \eqref{RTL_oim_m}. Here $D_n$ and $D^{-1}_n$ are shift operators acting according to rules $D_nu_n=u_{n+1}$ and $D^{-1}_nu_n=u_{n-1}$.

Let us rewrite equations \eqref{comp_cond_plus_0} and \eqref{comp_cond_minus_0} in a more detailed form:
\begin{equation}
\begin{aligned}\label{comp_cond_plus_1}
&\frac{dF}{dP_n}P_{n,t}+\frac{dF}{dQ_n}Q_{n,t}+\frac{dF}{du_n}u_{n,t}+\frac{dF}{dv_n}v_{n,t}+\frac{dF}{du_{n-1}}u_{n-1,t} \\
&\qquad \quad-\left(u_{n}P_{n}+(u_{n+1}-u_{n})P_{n+1}-u_{n+1}P_{n+2}+v_{n+1}(Q_{n+1}-Q_{n+2})\right)=0,\\
&\frac{dG}{dP_n}P_{n,t}+\frac{dG}{dQ_n}Q_{n,t}+\frac{dG}{du_n}u_{n,t}+\frac{dG}{dv_n}v_{n,t}+\frac{dG}{du_{n-1}}u_{n-1,t}-u_{n}(P_{n}-P_{n+1})=0,
\end{aligned} 
\end{equation}
\begin{equation}
\begin{aligned}\label{comp_cond_minus_1}
&\frac{df}{dP_n}P_{n,t}+\frac{df}{dQ_n}Q_{n,t}+\frac{df}{du_{n-1}}u_{n-1,t}+\frac{df}{dv_{n-1}}v_{n-1,t}+\frac{df}{du_{n-2}}u_{n-2,t} \\
&\qquad \quad-\left(u_{n-2}P_{n-2}+(u_{n-1}-u_{n-2})P_{n-1}-u_{n-1}P_{n}+v_{n-1}(Q_{n-1}-Q_{n})\right)=0,\\
&\frac{dg}{dP_n}P_{n,t}+\frac{dg}{dQ_n}Q_{n,t}+\frac{dg}{du_{n-1}}u_{n-1,t}+\frac{dg}{dv_{n-1}}v_{n-1,t}+\frac{dg}{du_{n-2}}u_{n-2,t}\\
&\qquad \quad -u_{n-2}(P_{n-2}-P_{n-1})=0,
\end{aligned} 
\end{equation}
mod (\eqref{RTL}, \eqref{RTL_lin1}, \eqref{RTL_oim1}, \eqref{RTL_oim_m}). We substitute variables $u_{n,t}$, $v_{n,t}$, $u_{n-1,t}$, $v_{n-1,t}$, $u_{n-2,t}$, $P_{n,t}$, $Q_{n,t}$, $P_{n+1,t}$, $Q_{n+1,t}$, $P_{n+2,t}$, $P_{n-1,t}$, $Q_{n-1,t}$, $P_{n-2,t}$ into equations \eqref{comp_cond_plus_1}, \eqref{comp_cond_minus_1} by virtue of equalities \eqref{RTL}, \eqref{RTL_lin1}, \eqref{RTL_oim1},  \eqref{RTL_oim_m}.
Then, for unknown functions $F$, $G$, $f$ and $g$ we obtain the following four equalities:
\begin{align}
&\left(u_{n-1}f+(u_n-u_{n-1})P_n-u_nF+v_n(Q_n-G)\right)\frac{dF}{dP_n} +u_{n-1}(f-P_{n})\frac{dF}{dQ_n} \nonumber \\
&\qquad \quad +u_n\left(u_{n-1}-u_{n+1}+v_n-v_{n+1}\right)\frac{dF}{du_n}+v_n\left(u_{n-1}-u_{n}\right)\frac{dF}{dv_n} \nonumber \\
&\qquad \quad +u_{n-1}\left(u_{n-2}-u_{n}+v_{n-1}-v_{n}\right)\frac{dF}{du_{n-1}} \nonumber \\
&\qquad \quad -\left(u_{n}P_{n}+(u_{n+1}-u_{n})F-u_{n+1}D_n(F)+v_{n+1}(G-D_n(G))\right)=0,\label{eq1}\\
&\left(u_{n-1}f+(u_n-u_{n-1})P_n-u_nF+v_n(Q_n-G)\right)\frac{dG}{dP_n} +u_{n-1}(f-P_{n})\frac{dG}{dQ_n}\nonumber \\
&\qquad \quad +u_n\left(u_{n-1}-u_{n+1}+v_n-v_{n+1}\right)\frac{dG}{du_n}+v_n\left(u_{n-1}-u_{n}\right)\frac{dG}{dv_n} \nonumber \\
&\qquad \quad +u_{n-1}\left(u_{n-2}-u_{n}+v_{n-1}-v_{n}\right)\frac{dG}{du_{n-1}}-u_{n}(P_{n}-F)=0, \label{eq2}\\
&\left(u_{n-1}f+(u_n-u_{n-1})P_n-u_nF+v_n(Q_n-G)\right)\frac{df}{dP_n} +u_{n-1}(f-P_{n})\frac{df}{dQ_n} \nonumber \\
&\qquad \quad +u_{n-1}\left(u_{n-2}-u_{n}+v_{n-1}-v_{n}\right)\frac{df}{du_{n-1}}+v_{n-1}\left(u_{n-2}-u_{n-1}\right)\frac{df}{dv_{n-1}} \nonumber \\
&\qquad \quad +u_{n-2}\left(u_{n-3}-u_{n-1}+v_{n-2}-v_{n-1}\right)\frac{df}{du_{n-2}} \nonumber \\
&\qquad \quad-\left(u_{n-2}D^{-1}_n(f)+(u_{n-1}-u_{n-2})f-u_{n-1}P_{n}+v_{n-1}(g-Q_{n})\right)=0, \label{eq3}\\
&\left(u_{n-1}f+(u_n-u_{n-1})P_n-u_nF+v_n(Q_n-G)\right)\frac{dg}{dP_n} +u_{n-1}(f-P_{n})\frac{dg}{dQ_n}\nonumber \\
&\qquad \quad +u_{n-1}\left(u_{n-2}-u_{n}+v_{n-1}-v_{n}\right)\frac{dg}{du_{n-1}}+v_{n-1}\left(u_{n-2}-u_{n-1}\right)\frac{dg}{dv_{n-1}}\nonumber \\
&\qquad \quad +u_{n-2}\left(u_{n-3}-u_{n-1}+v_{n-2}-v_{n-1}\right)\frac{dg}{du_{n-2}}-u_{n-2}(D^{-1}_n(f)-f)=0,\label{eq4}
\end{align} 
where
\begin{align*}
& D_n(F\left(P_n,Q_n,u_n,v_n,u_{n-1}\right))\\
&\qquad \qquad  =F\left(F\left(P_n,Q_n,u_n,v_n,u_{n-1}\right),G\left(P_n,Q_n,u_n,v_n,u_{n-1}\right),u_{n+1},v_{n+1},u_{n}\right), \\
& D_n(G\left(P_n,Q_n,u_n,v_n,u_{n-1}\right))\\
&\qquad \qquad  =G\left(F\left(P_n,Q_n,u_n,v_n,u_{n-1}\right),G\left(P_n,Q_n,u_n,v_n,u_{n-1}\right),u_{n+1},v_{n+1},u_{n}\right), \\
& D^{-1}_n(f\left(P_n,Q_n,u_{n-1},v_{n-1},u_{n-2}\right))\\
&\qquad \qquad  =f\left(f\left(P_n,Q_n,u_{n-1},v_{n-1},u_{n-2}\right),g\left(P_n,Q_n,u_{n-1},v_{n-1},u_{n-2}\right),u_{n-2},v_{n-2},u_{n-3}\right), \\
& D^{-1}_n(g\left(P_n,Q_n,u_{n-1},v_{n-1},u_{n-2}\right))\\
&\qquad \qquad  =g\left(f\left(P_n,Q_n,u_{n-1},v_{n-1},u_{n-2}\right),g\left(P_n,Q_n,u_{n-1},v_{n-1},u_{n-2}\right),u_{n-2},v_{n-2},u_{n-3}\right).
\end{align*}

Let us proceed directly to the study of the system of equations \eqref{eq1}-\eqref{eq4}, which is overdetermined. The scheme for solving the equations \eqref{eq1}-\eqref{eq4} is as follows. We will split them down according to the highest order shifts of the variables $u_n$ and $v_n$ and, thus, step by step we will solve the resulting equations. Moreover, at each step, we will refine the form of the functions $F$, $G$, $f$ and $g$ and substitute the obtained expressions again into the same equations \eqref{eq1}-\eqref {eq4}. Therefore, when we write in the future that we are working with the equations \eqref{eq1}-\eqref{eq4}, we mean that the previously found expressions for the sought functions have already been taken into account.

Differentiate the equation \eqref{eq3} with respect to the variable $u_{n-3}$ twice, and then from the resulting equation we find that the function $f$ depends linearly on the variable $u_{n-2}$:
\begin{align*}
f\left(P_n,Q_n,u_{n-1},v_{n-1},u_{n-2}\right)=f_1\left(P_n,Q_n,u_{n-1},v_{n-1}\right)u_{n-2}+f_2\left(P_n,Q_n,u_{n-1},v_{n-1}\right).
\end{align*}
Further, equating the coefficients at $u_{n-3}$ in the equations \eqref{eq3} and \eqref{eq4}, as well as at the variable $u_{n-2}$ in \eqref{eq1} and \eqref{eq2} we obtain that the functions $f$, $g$ and $F$, $G$ do not depend on the variables $u_{n-2}$ and $u_{n-1}$ respectively:
\begin{align*}
&f_1\left(P_n,Q_n,u_{n-1},v_{n-1}\right)=0, \\
&g\left(P_n,Q_n,u_{n-1},v_{n-1},u_{n-2}\right)=g_1\left(P_n,Q_n,u_{n-1},v_{n-1}\right),\\
&F\left(P_n,Q_n,u_n,v_n,u_{n-1}\right)=F_1\left(P_n,Q_n,u_n,v_n\right),\\
&G\left(P_n,Q_n,u_n,v_n,u_{n-1}\right)=G_1\left(P_n,Q_n,u_n,v_n\right).
\end{align*}

Let us equate the coefficients for the variables $u_{n+1}$, $v_{n+1}$ in the equation \eqref{eq2} and find:
\begin{align*}
G_1\left(P_n,Q_n,u_n,v_n\right)=G_2\left(P_n,Q_n,v_n\right).
\end{align*}
Next, we differentiate \eqref{eq4} with respect to $v_{n-2}$ and get that
\begin{align*}
f_2\left(P_n,Q_n,u_{n-1},v_{n-1}\right)=f_3\left(P_n,Q_n,u_{n-1}\right).
\end{align*}
On the next step, we twice differentiate equations \eqref{eq1} and \eqref{eq3} with respect to $u_{n+1}$ and $u_{n-2}$ respectively. From the equations obtained in this way, we find that functions $F_1$ and $f_3$ have the form 
\begin{align*}
&F_1\left(P_n,Q_n,u_n,v_n\right)=\frac{F_2\left(P_n,Q_n,v_n\right)}{u_n}+F_3\left(P_n,Q_n,v_n\right),\\
&f_3\left(P_n,Q_n,u_{n-1}\right)=\frac{f_4\left(P_n,Q_n\right)}{u_{n-1}}+f_5\left(P_n,Q_n\right).
\end{align*}
Note that in the equations \eqref{eq1} and \eqref{eq3} the dependence on the variables $u_{n+1}$ and $u_{n-2}$ respectively, is defined and therefore we can equate the coefficients of these variables. Then we get that the functions $F_3$ and $f_5$ are constant values: 
\begin{align*}
F_3=C_1, \qquad f_5=C_2.
\end{align*}
Let us act by the differentiation operator $\frac{\partial} {\partial u_{n-1}}$ on the equations \eqref{eq1} and \eqref{eq2}, by the operator $\frac{\partial}{\partial u_{n}}$  on the  \eqref{eq3} and \eqref{eq4}. As a result we obtain the following representations for the sought functions:
\begin{align*}
&F_2\left(P_n,Q_n,v_n\right)=\frac{F_4\left(Q_n-P_n,v_n(C_2-P_n)\right)}{C_2-P_n},\\
&G_2\left(P_n,Q_n,v_n\right)=G_3\left(Q_n-P_n,v_n(C_2-P_n)\right),\\
&f_4\left(P_n,Q_n\right)=\frac{f_6(Q_n)}{C_1-P_n},\\
&g_1\left(P_n,Q_n,u_{n-1},v_{n-1}\right)=g_2\left(Q_n,u_{n-1}(C_1-P_n),v_{n-1}\right).
\end{align*}
Let us collect the coefficients at the variable $u_n$ in \eqref{eq2} and obtain the following equation on the unknown function $G_3$:
\begin{align*}
(P_n-C_1)\left(\frac{dG_3\left(Q_n-P_n,v_n(C_2-P_n)\right)}{d(Q_n-P_n)}+1\right)+v_n(C_2-C_1)\frac{dG_3\left(Q_n-P_n,v_n(C_2-P_n)\right)}{d(v_n(C_2-P_n))}=0.
\end{align*}
Let us introduce in the last equality a temporary change of variables  
\begin{align*}
\theta_n=Q_n-P_n, \qquad \delta_n=v_n(C_2-P_n)
\end{align*}
and then the equation will be written in the form: 
\begin{align}\label{G2}
\left(\frac{dG_3(\theta_n,\delta_n)}{d\theta_n}+1\right)\left(P_n^2-(C_2+C_1)P_n+C_1C_2\right)+\delta_n(C_1-C_2)\frac{dG_3(\theta_n,\delta_n)}{d\delta_n}=0.
\end{align}
In the resulting equation, we compare the coefficients at $P^2_n$ and find:
\begin{align*}
G_3(\theta_n,\delta_n)=-\theta_n+G_4(\delta_n).
\end{align*}
Further, from \eqref{G2} we obtain two more relations:
\begin{align*}
\frac{dG_4(\delta_n)}{d\delta_n}=0 \qquad \textrm{or} \qquad C_2=C_1.
\end{align*}
The first relation after substitution into equation \eqref{eq2} leads to a contradiction; therefore, the second relation is true. As a result, returning to our variables, we have:
\begin{align*}
G_3(Q_n-P_n,v_n(C_1-P_n))=-Q_n+P_n+G_4(v_n(C_1-P_n)).
\end{align*}
In a similar way, we investigate the coefficients of $u_{n-2}$ in the equation \eqref{eq4} and finally find:
\begin{align*}
g_2\left(Q_n,u_{n-1}(C_1-P_n),v_{n-1}\right)=\frac{f_6(Q_n)}{u_{n-1}(C_1-P_n)}+g_3\left(Q_n,\frac{v_{n-1}}{u_{n-1}(C_1-P_n)}\right).
\end{align*}
Let us apply operator $\frac{\partial}{\partial v_n}$ to \eqref{eq3} and then by solving the resulting equation we obtain:
\begin{align*}
F_4\left(Q_n-P_n,v_n(C_1-P_n)\right)=v_n(C_1-P_n)(C_1+2Q_n-2P_n-G_4(v_n(C_1-P_n)))+F_5(Q_n-P_n).
\end{align*}
Then we subtract the equation \eqref{eq3} from the equation \eqref{eq4} so that the function $D_n^{-1}(f_6(Q_n))$ disappears. As a result, we get:
\begin{align*}
g_3\left(Q_n,\frac{v_{n-1}}{u_{n-1}(C_1-P_n)}\right)=&-Q_n+\frac{v_{n-1}Q^2_n}{u_{n-1}(C_1-P_n)}\\
&+\frac{v_{n-1}}{u_{n-1}(C_1-P_n)}g_4\left(-Q_n+\frac{u_{n-1}}{v_{n-1}}(C_1-P_n)\right).
\end{align*}
Further investigation of the equation \eqref{eq2} leads to the following relations:
\begin{align*}
&G_4(v_n(C_1-P_n))=C_3v_n(C_1-P_n)+C_4,\\
&F_5(Q_n-P_n)=f_6(Q_n)+\frac{(C_1-P_n)(C_4-2Q_n+P_n)}{C_3},
\end{align*}
where $C_3$, $C_4$ are arbitrary constants, such that $C_3\neq 0$. Case $C_3=0$ leads to a contradiction.

Finally, using equations \eqref{eq1} and \eqref{eq3} we find:
\begin{align*}
&f_6(Q_n)=\frac{(C_1+C_4)Q_n-Q_n^2}{C_3}+C_5,\\
&g_4\left(-Q_n+\frac{u_{n-1}}{v_{n-1}}(C_1-P_n)\right)=(C_1+C_4)\left(-Q_n+\frac{u_{n-1}}{v_{n-1}}(C_1-P_n)\right)-C_3C_5.
\end{align*}
Summarizing the above calculations, we find that the sought functions $F$, $G$, $f$ and $g$ have the form:
\begin{align*}
&F(P_n,Q_n,u_n,v_n,u_{n-1})=C_1+\frac{C_5}{u_n(C_1-P_n)}+\frac{v_n(C_1-C_4+2Q_n-2P_n)}{u_n}\\
&\qquad \qquad \qquad \qquad -\frac{C_3v_n^2(C_1-P_n)}{u_n}+\frac{(C_4+P_n-Q_n)(C_1-P_n+Q_n)}{u_nC_3(C_1-P_n)},\\
&G(P_n,Q_n,u_n,v_n,u_{n-1})=-Q_n+P_n+C_3v_n(C_1-P_n)+C_4,\\
&f(P_n,Q_n,u_{n-1},v_{n-1},u_{n-2})=C_1+\frac{C_5}{u_{n-1}(C_1-P_n)}+\frac{Q_n(C_1+C_4-Q_n)}{C_3u_{n-1}(C_1-P_n)},\\
&g(P_n,Q_n,u_{n-1},v_{n-1},u_{n-2})=C_1+C_4-Q_n-\frac{(C_3v_{n-1}-1)(C_1Q_n+C_4Q_n-Q_n^2+C_3C_5)}{C_3u_{n-1}(C_1-P_n)}.
\end{align*}
It is easy to show that the constants $C_1$ and $C_4$ are removed by point transformations, so we set $C_1=C_4=0$.
We denote $\lambda=\frac{1}{C_3}$, $C=C_5$ and then get the sought generalized invariant manifold \eqref{RTL_oim}.


\begin{thebibliography}{5}

\bibitem{Ruijsenaars} Ruijsenaars SNM. Relativistic Toda systems. Commun. Math. Phys. 133, 217--247, 1990.

\bibitem{Arutyunov} Arutyunov G. Elements of classical and quantum integrable systems, (UNITEXT for Physics) 1st ed. 2019 Edition, Kindle Edition, ISBN-13: 978-3030241971, ISBN-10: 3030241971.

\bibitem{Veselov2} Veselov AP. Integrable maps, Russian Math. Surveys, 46:5, 1--51, 1991.

\bibitem{Joshi} Gubbiotti G., Joshi N., Tran D., Viallet C. Bi-rational maps in four dimensions with two invariants. J. Phys. A: Math. Theor., 53:11, Art. 115201 -- 24. 2020.

\bibitem{Nijhoff} Xu X., Cao C., Nijhoff FW. Algebro-geometric integration of the Q1 lattice equation via nonlinear integrable
symplectic maps. Nonlinearity, 34, 2897–2918, 2021. 

\bibitem{Krichever85} Krichever IM. Nonlinear equations and elliptic curves. J. Soviet Math. 28:1, 51--90, 1985.

\bibitem{MEKL} Miller PD, Ercolani NM, Krichever IM, Levermore CD. Finite genus solutions to the Ablowitz-Ladik
equations Commun. Pure Appl. Math. 48, 1369--440, 1996.

\bibitem{Yanenko} Sidorov AF, Shapeev VP, Yanenko NN. Differential coupling method and its applications in gas dynamics. M.: Nauka 1984.

\bibitem{Sergeyev} Sergyeyev A. Constructing conditionally integrable evolution systems in (1 + 1) dimensions: a generalization of invariant modules approach. J. Phys. A: Math. Gen. 35, 7653--7660, 2002.

\bibitem{Demina} Demina MV. Classifying algebraic invariants and algebraically invariant solutions, Chaos, Solitons and Fractals, 140, 110219, 2020.

\bibitem{HKP2016JPA} Habibullin IT, Khakimova AR, Poptsova MN. On a method for constructing the Lax pairs for nonlinear integrable equations. J. Phys. A: Math. Theor. 49:3, 35 pp., 2016.

\bibitem{HK2020JPA} Habibullin IT, Khakimova AR. Invariant manifolds and separation of the variables for integrable chains. J. Phys. A: Math. Theor. 53:39, 25 pp., 2020.

\bibitem{Zhao} Zhao P, Fan E, Hou Yu. Algebro-geometric solutions for the Ruijsenaars–Toda hierarchy. Chaos, Solitons and Fractals. 54, 8--25, 2013. 

\bibitem{HKS2021UMJ} Habibullin IT, Khakimova AR, Smirnov AO. Generalized invariant manifolds for integrable equations and their applications. Ufa Math. J. 13:2, 141--157, 2021.

\bibitem{MikhailovShabatSokolov} Mikhailov AV, Shabat AB, Sokolov VV.  The symmetry approach to classification of integrable equations. What is integrability?  Springer, Berlin, Heidelberg. 115--184, 1991.

\bibitem{DubrovinMatveevNovikov} Dubrovin BA, Matveev VB, Novikov SP. Non-linear equations of Korteweg-de Vries type, finite-zone linear operators, and Abelian varieties. Russian Math. Surveys. 31:1, 59--146, 1976.

\bibitem{Volterra} Volterra V. Le$c_{,}$ons sur la th$\acute{e}$orie math$\acute{e}$matique de la lutte pour la vie. Paris: Gauthier-Villars. 1931. 

\bibitem{Veselov} Veselov AP. Integration of the stationary problem for a classical spin chain. Theor. Math. Phys. 71, 446--450, 1987.

\bibitem{Ver} Vereshchagin VL. Hamiltonian structure of averaged difference systems. Math. Notes. 44, 798--805, 1988.

\bibitem{krichever} Dubrovin BA, Krichever IM, Novikov SP. Integrable systems. I. (Russian) Current problems in mathematics. Fundamental directions. Itogi Nauki i Tekhniki, Akad. Nauk SSSR, Vsesoyuz. Inst. Nauchn. i Tekhn. Inform., Moscow.  4:291, 179--284, 1985.

\bibitem{SmirnovAO} Smirnov AO. Finite-gap solutions of Abelian Toda chain of genus 4 and 5 in elliptic functions. Theoret. and Math. Phys. 78:1, 6--13, 1989.

\bibitem{Smirnov20} Smirnov AO. Finite-gap solutions of the Toda lattice hierarchy. Modeling and situational quality management of complex systems. DOI: 10.31799/978-5-8088-1449-3-2020-1-41-44.

\bibitem{Khasanov} Babajanov BA, Khasanov AB. Integration of equation of Toda periodic chain kind. Ufa Math. J. 9:2, 17--24, 2017.

\bibitem{GongGeng} Gong D, Geng X. Quasi-periodic solutions of the relativistic Toda hierarchy. Journal of nonlinear mathematical physics. 19:4, 1250030, 35 pp, 2012. 

\bibitem{Zhdanov} Zhdanov RZ. Conditional Lie-B$\ddot{a}$cklund symmetries and reductions of evolution equations.
J. Phys. A: Math. Gen. 28, 3841--50, 1995.

\bibitem{FokasLiu} Fokas AS, Liu QM. Nonlinear interaction of traveling waves of nonintegrable equations.
Phys. Rev. Lett. 72, 3293--6, 1994.

\bibitem{ZYZhang} Zhang Zhi-Yong. An upper order bound of the invariant manifold in Lax pairs of a nonlinear evolution partial differential equation. J. Phys. A: Math. Theor. 52, 265202 (15pp), 2019. 

\bibitem{Belokolos} Belokolos ED, Bobenko AI, Enolskii VZ, Its AR, Matveev VB. Algebro-geometric approach to nonlinear integrable equations. Springer. 1994.

\bibitem{ItsMatveev} Its AR, Matveev VB. The Schr$\ddot{o}$dinger operator in a finite-zone spectrum and N-soliton solutions of the Korteweg-de Vries equation. Theor. Mat. Fiz. 9:51, 1975.

\bibitem{Batchenko} Batchenko V, Gesztesy F. On the spectrum of Schr$\ddot{o}$dinger operators with quasi-periodic algebro-geometric KDV potentials. Spinger. 2007.

\bibitem{KacMoerbeke} Kac M, van Moerbeke P. A complete solution of the periodic Toda problem. Proc. Nat. Acad. Sci. USA. 72:2879, 1975.

\bibitem{ZhouJiang} Zhou RG, Jiang QY. A Darboux transformation and an exact solution for the relativistic Toda lattice equation. J. Phys. A: Math. Gen. 38:735, 2005.

\bibitem{GengWuCao} Geng XG, Wu YT, Cao CW. Quasi-periodic solutions of the modified Kadomtsev-Petviashvili equation. J. Phys. A: Math. Gen. 32:3733, 1999.

\bibitem{CaoGengWang} Cao CW, Geng XG, Wang HY. Algebro-geometric solutions of the 2+1 dimensional Burgers equation with a discrete variable. J. Math. Phys. 43:621, 2002.

\bibitem{Zakharov} Novikov SP, Manakov SV, Pitaevskii LP, Zakharov VE. Theory of solitons the inverse scattering methods. New York: Consultants Bureau. 1984.

\bibitem{Olver} Olver PJ. Applications of Lie groups to differential equations. Bull. Amer. Math. Soc. 1988.

\end{thebibliography}
\end{document}